\documentclass[11pt, a4paper]{article}

\usepackage{microtype}

\usepackage{thmtools, thm-restate}

\usepackage[utf8]{inputenc} 
\usepackage[T1]{fontenc} 
\usepackage{geometry} 
\usepackage[%
    colorlinks      = true,
  linktocpage     = true,
  linkcolor       = blue,
  citecolor       = blue,
  urlcolor        = blue,
  breaklinks =true,
  hypertexnames=false,
]{hyperref}

\usepackage{setspace}
\usepackage[affil-it]{authblk}

\usepackage{amsmath, amsthm, amsfonts, amssymb, amsbsy, bigstrut, graphicx, enumerate,  upref, longtable, comment, booktabs, array, caption, subcaption}

\usepackage[numbers, sort&compress]{natbib}

\newcommand{\ep}{\epsilon}

\newtheorem{thm}{Theorem}[section]
\newtheorem{lmm}[thm]{Lemma}
\newtheorem{cor}[thm]{Corollary}
\newtheorem{prop}[thm]{Proposition}

\theoremstyle{definition}
\newtheorem{remark}[thm]{Remark}

\newcommand{\tr}{\operatorname{Tr}}

\newcommand{\uone}{\mathrm{U}(1)}
\newcommand{\inn}[1]{\langle #1 \rangle}

\numberwithin{equation}{section}

\renewcommand{\Re}{\operatorname{Re}}

\usepackage[utf8]{inputenc}
\usepackage{amsmath}
\usepackage{amsfonts}
\usepackage{bbm}
\usepackage{mathtools}
\usepackage{tikz}
\usetikzlibrary{decorations.pathreplacing}
\usepackage{amsthm}
\usepackage[english]{babel}
\usepackage{fancyhdr}
\usepackage{amssymb}
\usepackage{enumitem}
\usepackage{qtree}
\usepackage{mathrsfs}

\newcommand{\Z}{\mathbb{Z}}

\newcommand{\R}{\mathbb{R}}

\newcommand{\C}{\mathbb{C}}

\newcommand{\E}{\mathbb{E}}

\renewcommand{\tilde}{\widetilde}

\renewcommand{\i}{\mathrm{i}}
\renewcommand{\Im}{\operatorname{Im}}

\newcommand{\U}{\mathrm{U}}

\begin{document}
\title{A short proof of confinement in three-dimensional lattice gauge theories with a central $\mathrm{U}(1)$}
\author{Sourav Chatterjee\thanks{Department of Statistics, Stanford University, USA. Email: \href{mailto:souravc@stanford.edu}{\tt souravc@stanford.edu}. 
}}
\affil{Stanford University}


\maketitle

\begin{center}
{\it Dedicated to Len Gross on the occasion of his 95$^{\,th}$ birthday.}
\end{center}

\begin{abstract}
Pure lattice gauge theories in three dimensions are widely expected to confine.  A
rigorous proof of confinement for three-dimensional $\mathrm{U}(1)$ lattice gauge theory with Villain
action was given by G{\"o}pfert and Mack.  Beyond the abelian case, rigorous confinement
results are comparatively scarce; one general mechanism applies when the gauge group
has a central copy of $\mathrm{U}(1)$.  Indeed, combining a comparison inequality of Fr{\"o}hlich
with earlier work of Glimm and Jaffe yields confinement with a logarithmically growing
quark--antiquark potential for this class of theories.

The purpose of this note is to give a short, self-contained proof of this classical result for three-dimensional Wilson lattice gauge theory: when $G\subseteq \mathrm{U}(n)$
contains the full circle of scalar matrices $\{zI:\ |z|=1\}$, rectangular Wilson loops
obey an explicit upper bound of the form
$\lvert\langle W_\ell\rangle\rvert \le n\exp\{-c(1+n\beta)^{-1}T\log(R+1)\}$.
\end{abstract}


\section{Lattice gauge theories and Wilson loops}
\label{sec:model}
Let $G$ be a compact Lie subgroup of $\U(n)$ for some $n$.  Fix a finite set
$\Lambda\subseteq \mathbb{Z}^d$.  A (free-boundary) lattice gauge configuration on
$\Lambda$ assigns to each oriented nearest-neighbor edge $(x,y)$ with $x,y\in\Lambda$
an element $U(x,y)\in G$, subject to the constraint $U(y,x)=U(x,y)^{-1}$.
Let $G(\Lambda)$ denote the space of such configurations.

A unit square in $\Lambda$ bounded by four edges is called a plaquette; write
$P(\Lambda)$ for the set of plaquettes.  For $p\in P(\Lambda)$ with vertices
$x_1,x_2,x_3,x_4$ listed in counterclockwise order (with $x_1$ the lexicographically
smallest vertex and $x_2$ the second-smallest), define the plaquette holonomy
\[
U_p := U(x_1,x_2)\,U(x_2,x_3)\,U(x_3,x_4)\,U(x_4,x_1).
\]
The Wilson action of $U$ is defined as 
\begin{align*}
S_\Lambda(U) := \sum_{p\in P(\Lambda)} \Re(\tr(I-U_p)),
\end{align*}
where $I$ is the identity matrix of order $n$. 
Let $\sigma_{\Lambda}$ be the product Haar measure on $G(\Lambda)$, normalized to be a probability measure. Given $\beta >0$, let $\mu_{\Lambda, \beta}$ be the probability measure on $G(\Lambda)$ defined as
\[
d\mu_{\Lambda, \beta}(U) := \frac{1}{Z}e^{-\beta S_\Lambda(U)} d\sigma_{\Lambda}(U)\, ,
\]
where  $Z$ is the normalizing constant. This probability measure is called the Wilsonian lattice gauge theory on $\Lambda$ for the gauge group $G$, with inverse coupling strength $\beta$, and free boundary condition.

A Wilson loop $\ell$ is a sequence of points in $x_0,x_1,\ldots,x_k$ in $\Z^d$, not necessarily distinct, such that for each $0\le j\le k-1$, $x_{j+1}$ is a neighbor of $x_j$, and $x_0$ is a neighbor of $x_k$. In other words, $\ell$ is a closed loop in $\Z^d$. Suppose that a Wilson loop $\ell = (x_0,x_1,\ldots,x_k)$ is contained in $\Lambda$. The Wilson loop observable for $\ell$ is defined as
\[
W_\ell(U):= \tr(U(x_0,x_1)U(x_1,x_2)\cdots U(x_{k-1},x_k) U(x_k,x_0)). 
\]

\section{Quark confinement}
A central qualitative prediction of lattice gauge theory is confinement~\cite{wilson74}:
the energy required to separate a static quark--antiquark pair should diverge with the
separation.  In the lattice formulation this is often expressed through bounds on
rectangular Wilson loops. We say that the theory \emph{confines} if there exists a function $V:(0,\infty)\to\mathbb{R}$
with $V(R)\to\infty$ as $R\to\infty$ such that for every rectangular loop $\ell$ in
$\mathbb{Z}^d$ with side lengths $R\le T$,
\[
\lvert\langle W_\ell\rangle\rvert \le Ce^{-V(R)T},
\]
where $C$ is a constant independent of $R$ and $T$ (and the volume is taken large
enough so that $\ell$ is contained in~$\Lambda$).  A theory that does not satisfy such a
bound is called \emph{deconfining}. A particularly important case is \emph{Wilson's area law}, where $V(R)$ grows
linearly in $R$; equivalently, $\lvert\langle W_\ell\rangle\rvert$ decays at least like
$\exp\{-c\, \mathrm{area}(\ell)\}$ for some $c>0$.  By contrast, deconfining theories
typically exhibit a \emph{perimeter law}, with $\langle W_\ell\rangle$ decaying on the
scale of~$C_1 e^{-C_2(R+T)}$.

\section{Confinement in 3D lattice gauge theories}
This note concerns three-dimensional Wilson lattice gauge theory with a compact gauge
group $G\subseteq \U(n)$ whose center contains a full copy of $\U(1)$, namely
\[
\{zI:\ z\in\mathbb{C},\ |z|=1\}\subseteq G.
\]
This class includes $\U(n)$ itself and many related compact groups.  Our goal here is
expository: we give an elementary, self-contained argument implying confinement in the
sense that the effective potential $V(R)$ diverges with the separation $R$.
The resulting bound is logarithmic in $R$ (and therefore weaker than the expected
linear behavior in $d=3$), but the proof is short and stays entirely within the lattice
framework.  Our main result is the following.



\begin{thm}\label{confine3d}
Let $G$ be a compact Lie subgroup of $\U(n)$ for some $n$. Assume that $zI \in G$ for all $z\in \C$ with $|z|=1$, where $I$ denotes the $n\times n$ identity matrix. Let $\Lambda$ be a finite subset of $\Z^3$ and $\beta$ be a positive real number. Consider the lattice gauge theory on $\Lambda$ with this $G$ and $\beta$, under free boundary condition. For any rectangular loop $\ell$ contained in $\Lambda$ with side lengths $R\le T$, we have 
\[
|\inn{W_\ell}| \le n e^{-C (1+n\beta)^{-1} T \log (R+1)},
\]
where $C$ is a  positive universal constant.
\end{thm}

\paragraph{History.}
Using an inequality of~\citet*{messageretal78}, \citet{frohlich79} proved that $\U(n)$ lattice gauge theory confines whenever the
corresponding $\U(1)$ theory does.  Earlier,~\citet{glimmjaffe77}
established confinement for three-dimensional $\U(1)$ gauge theory with a
logarithmically growing potential, using ideas originating in the McBryan--Spencer
approach to the Mermin--Wagner theorem~\cite{mcbryanspencer77}.  Combining these results
yields Theorem~\ref{confine3d}; indeed it appears as Corollary~2 in~\cite{frohlich79}. The purpose of the present note is to provide a single, streamlined proof in this setting, with all estimates carried out directly on the lattice, and without appealing
to several separate arguments across multiple papers.

\paragraph{Idea of the proof.}
The proof of Theorem~\ref{confine3d} adapts the Mermin--Wagner ``complex rotation''
mechanism for the two-dimensional XY model~\cite{mermin67, merminwagner66}.  The central $\U(1)$
allows us to introduce auxiliary $\U(1)$ edge variables and to decompose the Wilson loop
into a $\U(1)$ phase factor and a bounded $G$-matrix factor; we then condition on the
latter.  A convenient geometric choice of conditioning yields an essentially
two-dimensional structure and a factorization over $T$ slices.  This reduces the
problem to bounding a two-point function in a weighted XY-type model in two
dimensions, for which we obtain a power-law bound by a short argument.

\paragraph{On the stronger area law.}
For $\U(1)$ gauge theory in $d=3$, confinement is known to be substantially stronger:
G{\"o}pfert and Mack~\cite{gopfertmack82} proved positivity of the string tension for all
couplings in a Villain-type formulation, i.e., an area law and at least linear growth
of the quark--antiquark potential.  It is natural to ask whether an appropriate version
of Fr{\"o}hlich's comparison inequality could be combined with such inputs to yield an
area law for three-dimensional $\U(n)$ gauge theory with Wilson action.  We do not
pursue this here, since doing so carefully would require additional technology and a
careful comparison between formulations and limits.

\section{Proof}\label{sec:logbound}
The proof requires the following `anti-concentration' lemma about a certain class of probability measures on $\uone$. 
\begin{lmm}\label{spreadlmm}
Take any $w\in \C$ and let $\mu$ be the probability measure on $\uone$ that has density proportional to $ \exp\{\Re(z\overline{w})\}$ with respect to the Haar measure on $\uone$. Then 
\[
\iint |z_1 - z_2|^2 d\mu(z_1) d\mu(z_2) \ge C \min\{1,|w|^{-1}\},
\]
where $C$ is a positive universal constant.
\end{lmm}
\begin{proof}
By Taylor expansion, we have that for any $x\in \R$,
\begin{align}\label{eynorm}
|e^{\i x} - 1 - \i x| &\le \frac{x^2}{2}.
\end{align}
Thus, for any $z\in \uone$ and $x\in \R$, 
\begin{align*}
|z - z e^{ \i x}| &\ge |\i z x|  - |z(1-e^{\i x} + \i x)|\\
&\ge |x| - \frac{1}{2}x^2.
\end{align*}
Choosing 
\[
\ep := \min\{1,|w|^{-1/2}\}\le 1, 
\]
we get that for any $z\in \uone$,
\begin{align}\label{mainep}
|z - z e^{\pm \i \ep}| &\ge \ep  - \frac{1}{2}\ep^2 \ge \frac{\ep}{2}.
\end{align}
Define $\rho:\uone \to [0,\infty)$ as 
\[
\rho(z) := \exp\{\Re(z\overline{w})\},
\]
and let 
\[
Z:= \int_{\uone} \rho(z) d\sigma_0(z),
\]
where $\sigma_0$ is the normalized Haar measure on $\uone$. Then note that $\mu$ has density $f= Z^{-1}\rho$ with respect to $\sigma_0$. 
Define $\rho^+(z) := \rho(z e^{\i\ep})$ and $\rho^-(z) := \rho(z e^{-\i \ep})$. Let $f^+$ and $f^-$ be the probability densities proportional to $\rho^+$ and $\rho^-$, respectively. By the invariance of Haar measure, the respective normalizing constants are both equal to $Z$. Let $\mu^+$ and $\mu^-$ denote the respective probability measures. Then note that for any measurable set $A\subseteq \uone$, the invariance of Haar measure implies that 
\begin{align}\label{muplus}
\mu^{\pm}(A) &= \frac{1}{Z}\int_A \rho(z e^{\pm\i \ep}) d\sigma_0(z)\notag \\
&= \frac{1}{Z}\int_{Ae^{\pm\i \ep}} \rho(z') d\sigma_0(z') = \mu(A e^{\pm \i \ep}).
\end{align}
Next, note that for any $z\in \uone$,
\begin{align*}
\sqrt{\rho^+(z)\rho^-(z)} &= \exp\biggl\{\frac{1}{2}\Re(\overline{w}z e^{\i \ep} + \overline{w} z e^{-\i \ep})\biggr\}\\
&= \rho(z) \exp\biggl\{\frac{1}{2}\Re(\overline{w}z(e^{\i \ep} + e^{-\i \ep} - 2))\biggr\}\\
&\ge \rho(z) \exp\biggl\{-\frac{1}{2}|w| |e^{\i \ep} + e^{-\i \ep} - 2|\biggr\}.
\end{align*}
But, by the inequality \eqref{eynorm}, 
\begin{align*}
|e^{\i \ep} + e^{-\i \ep} - 2| &\le |e^{\i \ep} - 1 -\i \ep| + |e^{-\i \ep} -1 +\i \ep|\le \ep^2.
\end{align*}
Thus, we get
\begin{align*}
\sqrt{\rho^+(z)\rho^-(z)} &\ge \rho(z) e^{-\frac{1}{2}|w|\ep^2}\ge \rho(z)e^{-1/2},
\end{align*}
where the last inequality holds because $\ep \le |w|^{-1/2}$. 
By the Cauchy--Schwarz inequality, this gives  for any measurable set $A\subseteq \uone$,
\begin{align*}
\sqrt{\mu^+(A) \mu^-(A)} &= \biggl(\int_A f^+(z) d\sigma_0(z) \int_A f^-(z)d\sigma_0(z)\biggr)^{1/2}\\
&\ge \int_A \sqrt{f^+(z) f^{-}(z)} d\sigma_0(z)\\
&= \frac{1}{Z} \int_A \sqrt{\rho^+(z) \rho^{-}(z)} d\sigma_0(z) \\
&\ge \frac{1}{Z}\int_A \rho(z) e^{-1/2} d\sigma_0(z) = \mu(A) e^{-1/2}. 
\end{align*}
Take any $z\in \uone$ and let $A := \{z'\in \uone: |z-z'| <  \frac{1}{4}\ep\}$. The above inequality shows that at least one $\mu^+(A)$ and $\mu^-(A)$ is $\ge \mu(A) e^{-1/2}$. First, suppose that $\mu^+(A) \ge \mu(A) e^{-1/2}$. Then by equation \eqref{muplus}, 
\begin{align*}
\mu(A \cap A e^{\i \ep}) &\ge \mu(A) + \mu(A e^{\i \ep }) - 1\\
&= \mu(A) + \mu^+(A) - 1\ge \mu(A) (1+ e^{-1/2}) -1.
\end{align*}
We claim that $A \cap A e^{\i \ep}$ is empty. To see this, notice that if $z' \in A \cap A e^{\i \ep}$, then $|z' - z|< \frac{1}{4}\ep$ and $|z' e^{-\i \ep } - z |< \frac{1}{4} \ep$. Thus, $|z' - z' e^{-\i \ep } | < \frac{1}{2} \ep $. But this contradicts the inequality \eqref{mainep}. So, $A \cap A e^{\i \ep}$ must be empty. Thus, the above inequality implies that 
\[
\mu(A) \le \frac{1}{1+ e^{-1/2}}. 
\]
The same argument shows that this inequality also holds if $\mu^-(A) \ge \mu(A) e^{-1/2}$. This leads to the conclusion that for any $z\in \uone$,
\begin{align*}
\int |z - z'|^2 d\mu(z') &\ge \frac{1}{16} \ep^2 \mu\biggl(\biggl\{z'\in \uone: |z - z'| \ge \frac{1}{4}\ep\biggr\}\biggr)\\
&\ge \frac{1}{16} \ep^2 \biggl(1 - \frac{1}{1+ e^{-1/2}}\biggr)\\
&= \frac{e^{-1/2}}{16(1+e^{-1/2})}\min\{1,|w|^{-1}\}.
\end{align*}
Integrating over $z$ completes the proof.
\end{proof}

Lemma \ref{spreadlmm} yields the following corollary.
\begin{cor}\label{spreadcor}
Let $\xi$ be a $\uone$-valued random variable following the probability density $f$ defined in Lemma \ref{spreadlmm}. Then 
\[
|\E(\xi)|\le 1 - C\min\{1,|w|^{-1}\},
\]
where $C$ is a positive universal constant.
\end{cor}
\begin{proof}
Let $\xi_1, \xi_2$ be i.i.d.~copies of $\xi$. Then 
\begin{align*}
|\xi_1 - \xi_2 |^2 &= 2 - 2\Re(\xi_1 \overline{\xi}_2).
\end{align*}
This gives
\begin{align*}
\E(|\xi_1  -\xi_2|^2) &= 2 - 2\Re(\E(\xi_1) \E(\overline{\xi}_2)) = 2 - 2|\E(\xi)|^2.
\end{align*}
Rearranging and applying Lemma \ref{spreadlmm}, we get
\[
|\E(\xi)|^2 = 1 - \frac{1}{2}\E(|\xi_1 -\xi_2 |^2) \le 1 - C\min\{1,|w|^{-1}\},
\]
where $C$ is a positive universal constant.  Without loss, $C\le 1$. Taking square-root on both sides and  applying the inequality $\sqrt{1-a}\le 1-\frac{1}{2}a$ (which holds for $a\in [0,1]$) completes the proof.
\end{proof}

We will now use Corollary \ref{spreadcor} to prove the following generalization of the classical Mermin--Wagner theorem~\cite{mermin67, merminwagner66}. In the following, `positively oriented edge' means a directed nearest-neighbor edge $(x,y)$ in $\Z^2$ where $x$ is less than $y$ in the lexicographic order.
\begin{lmm}\label{merminlmm}
Let $\Lambda$ be a finite subset of $\Z^2$ and let $E$ be the set positively oriented nearest-neighbor edges with both endpoints in $\Lambda$. Suppose that we are given a collection of complex numbers $(w_e)_{e\in E}$. Consider the probability measure $\gamma$ on $\uone^\Lambda$ that has density $\propto \exp(\sum_{e =(x,y)\in E} \Re(w_e \xi_x \overline{\xi}_y))$ with respect to the product of normalized Haar measures on $\uone^\Lambda$, at $\xi\in \uone^\Lambda$. Let $\phi$ be a random configuration drawn from $\gamma$. Then for any $x\in \Lambda$ and $R>0$ such that $y := x + (R,0)\in \Lambda$, we have 
\[
|\E(\overline{\phi}_x \phi_y)| \le e^{-CL^{-1}\log (R+1)},
\]
where $C$ is a positive universal constant and $L :=1+ \max_{e\in E} |w_e|$.
\end{lmm}
\begin{proof}
Fix $x$ and $R$. First, we note that there is no loss of generality in assuming that $\Lambda = x+\{-N, \ldots, N\}^2$ for some sufficiently large $N$. To see this, note that the original system can be `expanded' to a system of this form by defining $w_e = 0$ for the new edges. This just couples the old system with a collection of independent uniformly distributed $\phi_x$'s. Proving the required bound for the new system proves it for the old system.

For each $0\le k\le N$, let $y_k = x + (k,0)$ and let $L_k$ be the set of vertices in $\Lambda$ that are at $\ell^\infty$ distance equal to $k$ from $x$. Define a map $\tau: \uone^\Lambda\to \uone^\Lambda$ as follows. For $\xi \in \uone^\Lambda$, let $\tau(\xi)_x := \xi_x$, and for each $1\le k\le N$, let $\tau(\xi)_{y_k} := \xi_{y_k} \overline{\xi}_{y_{k-1}}$. Next, for each $1\le k\le N$ and $y\in L_k\setminus\{y_k\}$, let $\tau(\xi)_y := \xi_y \overline{\xi}_{y_k}$. 

Note that each $e\in E$ is either an edge connecting two neighboring vertices in some $L_k$, or connecting a vertex in some $L_k$ with a neighboring vertex in $L_{k-1}$. For each $1\le k\le N$, let $E_k$ denote the set of edges of the first kind, and let $F_k$ denote the set of edges of the second kind, except the edge $(y_{k-1},y_k)$. Further decompose $F_k$ as $F_k^1\cup F_k^2$, where $F_k^1$ is the set of all $(y,y')\in F_k$ with $y\in L_{k-1}$ and $y'\in L_k$, and $F_k^2$ is the set of all $(y,y')\in F_k$ with $y\in L_k$ and $y' \in L_{k-1}$.

Take any $\xi \in \uone^\Lambda$ and let $\chi := \tau(\xi)$. Then by the above observation,
\begin{align*}
f(\xi) := \sum_{e =(y,y')\in E} w_e \xi_y\overline{\xi}_{y'} &= \sum_{k=1}^{N} w_{(y_{k-1},y_{k})} \xi_{y_{k-1}} \overline{\xi}_{y_{k}} + \sum_{k=1}^N\sum_{e = (y,y')\in E_k} w_e \xi_y \overline{\xi}_{y'}\\
&\qquad + \sum_{k=1}^N\sum_{e = (y,y')\in F_k} w_e \xi_y \overline{\xi}_{y'}.
\end{align*}
Take any $2\le k\le N$. Note that  for $(y,y')\in F_k^1$, 
\begin{align*}
\xi_{y} \overline{\xi}_{y'}  &= \xi_{y}\overline{\xi}_{y_{k-1}} \overline{\xi}_{y'} \xi_{y_k}\overline{\xi}_{y_k}\xi_{y_{k-1}} = \chi_{y} \overline{\chi}_{y'}\overline{\chi}_{y_k}.
\end{align*}
Similarly, for any $(y,y')\in F_k^2$, 
\begin{align*}
\xi_{y} \overline{\xi}_{y'}  &= \xi_{y}\overline{\xi}_{y_{k}} \overline{\xi}_{y'} \xi_{y_{k-1}}\overline{\xi}_{y_{k-1}}\xi_{y_{k}} = \chi_{y} \overline{\chi}_{y'}\chi_{y_k}.
\end{align*}
When $k=1$ and $(y,y')\in F_k^1$,  we have $y = x = y_0$ and $y'\in L_1$, which gives 
\[
\xi_y \overline{\xi}_{y'} = \xi_{y_0}\overline{\xi}_{y_1}\overline{\xi}_{y'} \xi_{y_1} = \overline{\chi}_{y_1} \overline{\chi}_{y'}.
\]
Similarly, if $k=1$ and $(y,y')\in F_k^2$, then $y' = x=y_0$ and $y\in L_1$, which gives 
\[
\xi_y \overline{\xi}_{y'} = \xi_y \overline{\xi}_{y_1}\xi_{y_1}\overline{\xi}_{y_0} = \chi_y \chi_{y_1}.
\]
Next, take any $1\le k\le N$ and note that for $(y,y')\in E_k$ with $y,y'\ne y_k$,
\begin{align*}
\xi_{y} \overline{\xi}_{y'}  &= \xi_{y}\overline{\xi}_{y_{k}} \overline{\xi}_{y'} \xi_{y_k} = \chi_{y} \overline{\chi}_{y'}.
\end{align*}
On the other hand, if $(y,y')\in E_k$ and $y=y_k$, then 
\begin{align*}
\xi_{y} \overline{\xi}_{y'}  &= \overline{\xi}_{y'} \xi_{y_k} = \overline{\chi}_{y'},
\end{align*}
and if $(y,y')\in E_k$ and $y' = y_k$, then
\begin{align*}
\xi_{y} \overline{\xi}_{y'}  &=\xi_{y} \overline{\xi}_{y_k} = \chi_{y}.
\end{align*}
Combining the above observations, we get
\begin{align*}
f(\xi) &= \sum_{k=1}^{N} w_{(y_{k-1},y_{k})} \overline{\chi}_{y_{k}}+ \sum_{k=1}^N\sum_{\substack{e = (y,y')\in E_k, \\ y,y'\ne y_k}} w_e \chi_y \overline{\chi}_{y'}\\
&\qquad + \sum_{k=1}^N w_{(y_k, y_k+(0,1))}\overline{\chi}_{y_k+(0,1)} + \sum_{k=1}^N w_{(y_k - (0,1), y_k)} \chi_{y_k-(0,1)}\\
&\qquad + \sum_{k=2}^N\biggl(\sum_{e = (y,y')\in F_k^1} w_e \chi_y \overline{\chi}_{y'}\biggr)\overline{\chi}_{y_k} + \sum_{k=2}^N\biggl(\sum_{e = (y,y')\in F_k^2} w_e \chi_y \overline{\chi}_{y'}\biggr)\chi_{y_k} \\
&\qquad + \sum_{e=(y,y') \in F_1^1} w_e \overline{\chi}_{y'}\overline{\chi}_{y_1} + \sum_{e=(y,y') \in F_1^2} w_e \chi_{y}\chi_{y_1}=: g(\chi).
\end{align*}
Let $\phi$ be a random configuration with density proportional to $e^{\Re(f)}$, and let $\psi:= \tau(\phi)$. 
We claim that $\psi$ has density proportional to $e^{\Re(g)}$. To see this, let us first show that for any bounded measurable $F:\uone^\Lambda \to \C$,
\begin{align}\label{eref}
\int F(\tau(\xi))\prod_{y\in \Lambda} d\sigma_0(\xi_y) = \int F(\xi)\prod_{y\in \Lambda} d\sigma_0(\xi_y),
\end{align}
where $\sigma_0$ denotes the normalized Haar measure on $\uone$.  To prove this claim, we first integrate over $\xi_y$ for all $y\in \Lambda \setminus\{y_0,\ldots,y_N\}$. Take any $y\in L_k\setminus\{y_k\}$. In the integrand, $\xi_y$ appears only through $\tau(\xi)_y = \xi_y \overline{\xi}_{y_k}$. By the invariance of Haar measure, we can replace this simply by $\xi_y$. Next, we successively integrate over $\xi_{y_N}, \xi_{y_{N-1}},\ldots,\xi_{y_0}$, in this order. Note that $\xi_{y_N}$ appears in the integrand only through $\tau(\xi)_{y_N} = \xi_{y_N} \overline{\xi}_{y_{N-1}}$. Again, by the invariance of Haar measure, we can replace this by $\xi_{y_N}$. After this replacement, the integrand now depends on $\xi_{y_{N-1}}$ only through $\tau(\xi)_{y_{N-1}} = \xi_{y_{N-1}} \overline{\xi}_{y_{N-2}}$. Thus, we can replace this term by $\xi_{y_{N-1}}$, and so on. This proves the claim \eqref{eref}. 

Let us now prove that $\psi$ has density proportional to $e^{\Re(g)}$. Let $F$ be as above. Then by equation \eqref{eref} and the definition of $g$, 
\begin{align*}
\int F(\tau(\xi)) e^{\Re(f(\xi))} \prod_{y\in \Lambda} d\sigma_0(\xi_y)  &= \int F(\tau(\xi)) e^{\Re(g(\tau(\xi)))} \prod_{y\in \Lambda} d\sigma_0(\xi_y)\\
&= \int F(\xi) e^{\Re(g(\xi))} \prod_{y\in \Lambda} d\sigma_0(\xi_y).
\end{align*}
Since this holds also for $F\equiv 1$, we get
\begin{align*}
\E(F(\psi)) &=\frac{\int F(\tau(\xi)) e^{\Re(f(\xi))} \prod_{y\in \Lambda} d\sigma_0(\xi_y) }{\int e^{\Re(f(\xi))} \prod_{y\in \Lambda} d\sigma_0(\xi_y) }\\
&= \frac{\int F(\xi) e^{\Re(g(\xi))} \prod_{y\in \Lambda} d\sigma_0(\xi_y) }{\int e^{\Re(g(\xi))} \prod_{y\in \Lambda} d\sigma_0(\xi_y) }.
\end{align*}
This proves that $\psi$ has density proportional to $e^{\Re(g)}$. 

Let $\E'$ denote conditional expectation given $\{\psi(y):y\notin \{y_1,\ldots,y_N\}\}$. From the form of $g$, we see that under this conditioning, the random variables $\psi_{y_1},\ldots,\psi_{y_N}$ are independent, with $\psi_{y_k}$ having conditional density proportional to $\exp(\Re(\overline{u}_k\chi))$ at $\chi \in \uone$, where
\[
u_k := w_{(y_{k-1}, y_k)} + \sum_{e = (y,y')\in F_k^1} w_e \psi_y \overline{\psi}_{y'} + \sum_{e = (y,y')\in F_k^2} \overline{w}_e \overline{\psi}_y \psi_{y'}
\]
if $2\le k\le N$, and
\[
u_1 := w_{(y_{0}, y_1)} + \sum_{e = (y,y')\in F_1^1} w_e \overline{\psi}_{y'} + \sum_{e = (y,y')\in F_1^2} \overline{w}_e \overline{\psi}_y.
\]
Note that $|u_k|\le CL k$, where $C$ is a universal constant and $L = 1+\max_{e\in E} |w_e|$. Thus, by Corollary \ref{spreadcor},
\begin{align*}
|\E'(\psi_{y_k})| &\le 1 - C_0 \min\{1, (C_1L k)^{-1}\}
\end{align*}
for some  positive universal constants $C_0$ and $C_1$. Thus, for any $1\le R\le N$,
\begin{align*}
|\E'(\psi_{y_1}\cdots \psi_{y_R})| &= \prod_{k=1}^R |\E'(\psi_{y_k})|\le \prod_{k=1}^R (1 - C_0 \min\{1, (C_1L k)^{-1}\})\\
&\le \prod_{k=1}^R e^{ - C_0 \min\{1, (C_1L k)^{-1}\}} = \exp\biggl(-C_0\sum_{k=1}^R \min\{1,(C_1Lk)^{-1}\}\biggr). 
\end{align*}
Now, note that since $L\ge 1$, we have
\[
\min\{1, (C_1Lk)^{-1}\} \ge \min\{(Lk)^{-1}, (C_1Lk)^{-1}\} = C_2 (Lk)^{-1}, 
\]
where $C_2 := \min\{1, C_1^{-1}\}$. Thus, we get
\begin{align*}
|\E'(\psi_{y_1}\cdots \psi_{y_R})| &\le \exp\biggl(-\frac{C_0C_2}{L}\sum_{k=1}^R \frac{1}{k}\biggr)\le \exp\biggl(-\frac{C_0C_2}{L}\log (R+1)\biggr). 
\end{align*}
Thus, the unconditional expected value  $|\E(\psi_{y_1}\cdots \psi_{y_R})|$ is also bounded above by the same quantity. But $\psi_{y_1}\cdots \psi_{y_R} = \overline{\phi}_x \phi_{y_R}$. This completes the proof.
\end{proof}

We are now ready to prove Theorem \ref{confine3d}.
\begin{proof}[Proof of Theorem \ref{confine3d}]
First, we construct an `expanded' model that is equivalent to our lattice gauge theory on $\Lambda$ in a certain sense. In the expanded model, configurations consist of maps from $E(\Lambda)$ into $G\times \uone$. Let $(U,\xi)$ be such a configuration. For a negatively oriented edge $(x,y)$, we let $U(y,x):= U(x,y)^{-1}$ and $\xi(y,x):= \xi(x,y)^{-1}$.  For a plaquette $p$ with successive oriented edges $e_1,e_2,e_3,e_4$, we define $U_p$ as before and analogously let $\xi_p := \xi(e_1)\xi(e_2)\xi(e_3)\xi(e_4)$. Finally, we consider the probability measure 
\[
d\nu(U,\xi) = \frac{1}{\tilde{Z}}\exp\biggl(-\beta \sum_{p\in P(\Lambda)} \Re(\tr(I-\xi_pU_p))\biggr) \prod_{e\in E(\Lambda)} d\sigma(U(e)) d\sigma_0(\xi(e)), 
\]
where $\sigma$ and $\sigma_0$ are the normalized Haar measures on $G$ and $\uone$, respectively, and $\tilde{Z}$ is the normalizing constant.

We claim that $\tilde{Z} = Z$, where $Z$ is the normalizing constant of our original lattice gauge theory on $\Lambda$. To see this, recall that $zI\in G$ for all $z\in \C$ with $|z|=1$. Take any $\xi$. Since $|\xi(e)|=1$ for all $e$, and 
\[
\xi_p U_p = \xi(e_1) U(e_1) \xi(e_2)U(e_2)\xi(e_3)U(e_3) \xi(e_4)U(e_4)
\]
for a plaquette $p$ with successive oriented edges $e_1,e_2,e_3,e_4$, 
the invariance of $\sigma$ implies that 
\begin{align*}
\tilde{Z} &=\int \exp\biggl(-\beta \sum_{p\in P(\Lambda)} \Re(\tr(I-\xi_pU_p))\biggr) \prod_{e\in E(\Lambda)} d\sigma(U(e))\\
&= \int \exp\biggl(-\beta \sum_{p\in P(\Lambda)} \Re(\tr(I-U_p))\biggr) \prod_{e\in E(\Lambda)} d\sigma(U(e)) = Z. 
\end{align*}
Since this holds for any $\xi$ and $\sigma_0$ is a probability measure, we can now integrate over $\xi$ and conclude that $\tilde{Z} = Z$.

Next, for any $(U,\xi)$, let $\xi U\in G^{E(\Lambda)}$ denote the configuration $(\xi U)(e) := \xi(e) U(e)$. Then the same argument as above implies that for any bounded measurable function $f:G^{E(\Lambda)} \to \C$,
\begin{align*}
&\int f(\xi U) \exp\biggl(-\beta \sum_{p\in P(\Lambda)} \Re(\tr(I-\xi_pU_p))\biggr) \prod_{e\in E(\Lambda)} d\sigma(U(e))\\
&= \int f(U) \exp\biggl(-\beta \sum_{p\in P(\Lambda)} \Re(\tr(I-U_p))\biggr) \prod_{e\in E(\Lambda)} d\sigma(U(e)). 
\end{align*}
Combining, we get that 
\begin{align}\label{expansion}
\int f(\xi U) d\nu(U,\xi) = \int f(U) d\mu(U),
\end{align}
where $\mu$ is the probability measure for the original lattice gauge theory.

Henceforth, we will let $(U,\xi)$ denote a random configuration drawn from $\nu$, and $V$ be a random configuration drawn from $\mu$. Let $\ell$ be the rectangular loop with vertices $(0,0,0)$, $(R, 0, 0)$, $(0, T,0)$ and $(R,T,0)$. Let $e_1,\ldots, e_k$ be the successive oriented edges of $\ell$ (with some arbitrary $e_1$), and define
\begin{align*}
&W_\ell := \tr(V(e_1)\cdots V(e_k)),\\
&Q_\ell := \tr(U(e_1)\cdots U(e_k)),\\
&\chi_\ell := \xi(e_1)\cdots \xi(e_k).
\end{align*}
Then by equation \eqref{expansion},
\begin{align*}
\E(W_\ell) &= \E[\tr(\xi(e_1)U(e_1)\cdots \xi(e_k)U(e_k))]= \E(\chi_\ell Q_\ell).
\end{align*}
Let $\E'$ denote conditional expectation given $U$. Then by the above equation,
\[
\E(W_\ell) = \E(\E'(\chi_\ell) Q_\ell),
\]
which implies that 
\[
|\E(W_\ell)| \le \E(|\E'(\chi_\ell) Q_\ell|)\le n \E(|\E'(\chi_\ell)|). 
\]
Thus, it suffices to show that 
\begin{align}\label{chitoshow}
|\E'(\chi_\ell)|\le e^{-C(1+n\beta)^{-1}T \log (R+1)}
\end{align}
for some positive universal constants $C$.

Let $E'\subseteq E(\Lambda)$ denote the set of edges of the form $(x,x+e_2)$, where $e_2=(0,1,0)$. Let $E'' := E(\Lambda) \setminus E'$, and let $\E''$ denote conditional expectation given $U$ and $(\xi(e))_{e\in E''}$. We will prove that 
\begin{align}\label{chitoshow2}
|\E''(\chi_\ell)|\le e^{-C(1+n\beta)^{-1}T \log (R+1)}
\end{align}
for some positive universal constant $C$. Clearly, this implies the bound \eqref{chitoshow} and completes the proof. 

For $0\le j< T$, let $p_j$ be the edge $(j e_2, (j+1)e_2)$ and let $q_j$ be the edge $(Re_1 + je_2, Re_1+(j+1)e_2)$. Also, let $E_j\subseteq E(\Lambda)$ be the set of edges of the form $(x, x+e_2)$ where the second coordinate of $x$ is $j$.  Repeating the observation that we have already made several times earlier, note that conditional on $U$ and $(\xi(e))_{e\in E''}$, the collections $(\xi(e))_{e\in E_j}$, $j=0,1,\ldots,T-1$ are independent. Consequently,
\begin{align*}
\E''(\chi_\ell) &= \chi'_\ell\prod_{j=0}^{T-1} \E''(\overline{\xi(p_j)} \xi(q_j)),
\end{align*}
where $\chi'_\ell$ is the product of $\xi(e)$ over all $e$ in the lower boundary of $\ell$ and $\overline{\xi(e)}$ over all $e$ in the upper boundary of $\ell$. In particular,
\begin{align}\label{chitoshow3}
|\E''(\chi_\ell)| &= \prod_{j=0}^{T-1} |\E''(\overline{\xi(p_j)} \xi(q_j))|.
\end{align}
Now, it is easy to see that conditional on $U$ and $(\xi(e))_{e\in E''}$, the collection $(\xi(e))_{e\in E_j}$ is a 2D spin system, with $\uone$-valued spins, and conditional density of the form considered in Lemma \ref{merminlmm}.  The corresponding $w$'s are uniformly bounded by $C_0n\beta$, where $C_0$ is a positive universal constant. Thus, for any $j$,
\[
|\E''(\overline{\xi(p_j)} \xi(q_j))|\le e^{-C_1(1+n \beta)^{-1} \log (R+1)},
\]
where, again, $C_1$ is a positive universal constant. By equation \eqref{chitoshow3}, this proves the claim~\eqref{chitoshow2} and therefore, completes the proof of the theorem. 
\end{proof}
\vskip.3in

\noindent {\bf Acknowledgments.}  I thank Christophe Garban for helpful conversations.
\vskip.1in
\noindent {\bf Funding information.} This work was supported by NSF grant DMS-2450608 and the Simons Collaboration grant on `Probabilistic Paths to QFT'.
\vskip.1in
\noindent {\bf Author contribution.} Sourav Chatterjee is the sole author of this manuscript.
\vskip.1in
\noindent {\bf Conflict of Interest.} None.

\bibliographystyle{abbrvnat}

\bibliography{myrefs}

\end{document}